\documentclass[12pt]{article}

\usepackage[utf8]{inputenc}

\usepackage[english]{babel}

\usepackage{amsmath,amsthm,amssymb}

\usepackage[pdfencoding=unicode, psdextra, citecolor=blue, urlcolor=blue,
linkcolor=blue, colorlinks=true, bookmarksopen=true]{hyperref}

\usepackage[top=2cm,bottom=2cm,left=1.5cm,right=1.5cm]{geometry}

\usepackage{graphicx}

\usepackage{authblk}

\newcommand{\R}{\mathbb{R}}
\newcommand{\rr}{{\mathbf{r}}}
\newcommand{\D}{\mathrm{d}}
\newcommand{\Li}{\mathcal{L}}

\DeclareMathOperator{\sgn}{sgn}
\DeclareMathOperator{\tr}{tr}

\newtheorem{theorem}{Theorem}
\newtheorem{proposition}{Proposition}
\newtheorem{lemma}{Lemma}
\newtheorem{corollary}{Corollary}

\theoremstyle{definition}

\theoremstyle{remark}

\newtheorem{remark}{Remark}

\title{Time-optimal state transfer for an open qubit\footnote{This work was funded by the Ministry of Science and Higher Education of Russian Federation (grant number 075-15-2020-788)}}

\author[1,2]{Lokutsievskiy L.V.}
\author[1,3]{Pechen A.N.}
\author[2]{Zelikin M.I.}

\affil[1]{Steklov Mathematical Institute of Russian Academy of Sciences, 8 Gubkina str., Moscow 119991, Russia}

\affil[2]{Lomonosov Moscow State University, GSP-1, Leninskie Gory, Moscow 119991, Russia}

\affil[3]{National University of Science and Technology ”MISIS”, 4 Leninsky Prosp.,
Moscow 119991, Russia}

\begin{document}

\maketitle

\begin{abstract}
Finding minimal time and establishing the structure of the corresponding optimal controls which can  transfer a given initial state of a quantum system into a given target state is a key problem of quantum control. In this work, this problem is solved for a basic component of various quantum technology processes --- a qubit interacting with the environment and experiencing an arbitrary time-dependent coherent driving. We rigorously derive both upper and lower estimates for the minimal steering time. Surprisingly, we discover that the optimal controls have a very special form --- they consist of two impulses, at the beginning and at the end of the control period, which can be  assisted by a smooth time-dependent  control in between. Moreover, an important for practical applications explicit almost optimal state transfer protocol is provided which only consists of four impulses and gives an almost optimal time of motion. The results can be directly applied to a variety of experimental situations for estimation of the ultimate limits of state control for quantum technologies.
\end{abstract}

\section{Introduction}

Estimation of minimal time necessary for steering a quantum system from a given initial state to a given target state is among key problems for quantum technologies~\cite{Schleich_etal_2016,Acin_etal_2018,Koch_etal_2022}. An important example is a two-level quantum system, which serves as a practical model for a variety of experimental situations and basic model for quantum computation and information transmission. As an example, two-level systems describe single qubits for various experimental quantum computation platforms~\cite{Ohtsuki_Mikami_Ajiki_Tannor_2023}, they appear as spin 1/2 systems in Nuclear Magnetic Resonance (NMR)~\cite{Lapert_Zhang_Braun_Glaser_Sugny_2010,Zhang_Lapert_Sugny_Braun_Glaser_2011}, describe electronic excitations in light-harvesting systems~\cite{Fassioli_Dinshaw_Arpin_Scholes_2014}, as well as many other physical situations. In many such tasks, the problem of minimal time steering of an initial state to a target state naturally arises. For example, in quantum computation it is motivated by the needs of fast state initialization, in NMR  by the needs of fast preparation of the maximally mixed state of a spin in the surrounding environment~\cite{Lapert_Zhang_Braun_Glaser_Sugny_2010}, in photosynthesis it corresponds to minimization of the recombination time~\cite{Kozyrev_Pechen_2022} necessary for fast electron transport. For this reason, the analysis of time-optimal control of two-level systems attracts an extremely high interest. 

In some cases, the two-level system is considered as a closed system with coherent driving, whose evolution is described by the Schr\"odinger equation. In many situations however the assumption that the system is isolated from the environment is too strong and one has to consider it as an open quantum system. Time-optimal control for a two-level quantum system evolving according to a Gorini--Kossakowski--Lindblad--Sudarshan (GKSL) master equation modeling population relaxation and dephasing with bounded coherent control was studied in~\cite{Sugny_Kontz_Jauslin_2007} based on the fundamental work~\cite{BoscainPicolli12004} devoted to control systems on 2D manifolds. Optimal control at the quantum speed limit for the Landau-Zener two-level quantum system with variable magnetic field was studied in~\cite{Caneva_2009}. For general closed two-level systems with two or three drives, quantum speed limit was established in~\cite{Hegerfeldt_2013,Hegerfeldt_2014}, where it was found that optimal protocol consists of two $\delta$-like pulses and a period in between. Time-optimal protocol for driving a general initial state to a target state by a single control field with bounded amplitude was established in~\cite{Lin_Sels_Wang_2020}. Time-optimal universal control of Hamiltonian two-level systems under strong driving was investigated~\cite{Avinadav_2014}. Geometric quantum speed limits for Markovian dynamics with constant Hamiltonians were studied in~\cite{Lan_Xie_Cai_2022}. An interesting analysis was performed for one or two bounded piecewise constant controls using a suitable extension of the Pontryagin maximum principle~\cite{Dionis_Sugny_2023}. 

The key problem of estimating minimal time and establishing the structure of the corresponding optimal protocols for an open quantum system which evolves under a dissipative evolution with arbitrary time-dependent controls in the Hamiltonian has yet been remained unsolved. This work completely solves this problem allowing to obtain surprising results on the structure of the optimal controls. The solution requires a combination of special subtle techniques form geometric control theory. For this, we consider the dynamics of a two-level open quantum system driven by an arbitrary time-dependent control $u(t)$ as described by a dissipative GKSL master equation
\begin{equation}
\label{eq:master}
\frac{{\rm d} \rho}{{\rm d} t} = -\frac{{\rm i}}{\hbar}[H_0+u(t)V,\rho] + \Li(\rho).
\end{equation}
Here, $\rho$ is the $(2\times2)$ system density matrix, $H_0$ is the free Hamiltonian, $u(t)\in\R$ is the coherent control, $V$ is the coherent control Hamiltonian, and the GKSL superoperator $\Li$ described a non-unitary interaction of the system with its surrounding environment. The time optimal problem is to find a control $u(t)$ such that
\begin{equation}
\label{problem:master}
    \begin{gathered}
    \rho(0)=\rho^0;\quad\rho(T)=\rho^1;\\
        T\to\min.
    \end{gathered}
\end{equation}
Here $\rho^0$ and $\rho^1$ are two given initial and final density matrices, and $T\ge 0$ denotes the time necessary for steering $\rho^0$ into $\rho^1$. Without interaction with the environment, the system is obviously uncontrollable in the set of all density matrices, for example because it is impossible to convert any pure state into a mixed state solely through unitary transformations.

In~\cite{LokutsievskiyPechen2021}, it was shown that the two-level control system is uncontrollable even in the presence of incoherent control. However, in this work it was demonstrated that the system is ''almost'' controllable in the following sense. Let $\omega>0$ be the eigenfrequency of the system and $\gamma>0$ be the decoherence rate, $\gamma\ll\omega$. Then, any state $\rho^1$ with a purity level 
\[
    P(\rho)=\tr\rho^2<1-\frac\pi4\frac\gamma\omega
\]
can be reached from any other state $\rho^0$ in a finite time. Surprisingly, the structure of the reachable set does not change if incoherent control is added to the system~\cite{LokutsievskiyPechen2021}.

In this work, estimates for the minimal steering  time and the structure of the corresponding optimal and almost optimal controls are found. In details,
\begin{itemize}
\item We obtain upper and lower bounds on the optimal motion time $T$ depending on the initial and final states $\rho^0$ and $\rho^1$. For the initial and final states, let us denote $\mu_{0,1}=\sqrt{2\tr(\rho^{0,1})^2 - 1}\in[0;1]$, and $\sigma=\sgn\,(\mu^1-\mu^0)=\pm1$, then
\[
  \frac1\gamma\ln\left(\frac{1-\sigma \mu_0}{1-\sigma \mu_1}\right) \le
  T\le
  \frac{\pi}{\omega} + \frac{e}\gamma + 
   \frac1\gamma \ln\left(\frac{1-\sigma \mu_0}{1-\sigma \mu_1}\right).
\]
The lower bound holds when $\mu_1<1$, and the upper bound holds when $\mu_1\le 1-\frac\pi2\frac\gamma\omega$ (see Theorem~\ref{thm:time_estimates}).
 
\item We prove that optimal controls are in the class of impulse controls and show an example of the absence of optimal control in the class $L_1$ (Proposition~\ref{prop:auxiliary_existence}, Theorem~\ref{thm:times_original_auxiliary_equal}, and Remark~\ref{rm:not_L1_but_impulse}).
	
\item We prove that the optimal control is a sum of at most two impulses in the form of Dirac delta functions at the initial and final time moments and an analytic function in between (Theorem~\ref{thm:times_original_auxiliary_equal}).

\item We suggest an explicit form of an almost optimal control with four impulses implementing the upper bound on the optimal motion time $T$ (formulas~\eqref{eq:theta_suboptimal}, \eqref{eq:v_suboptimal} and~\eqref{eq:u_v_via_theta}).
\end{itemize}

Figure~\ref{Fig1} shows lower and upper estimates for the minimal time necessary for moving between sets with purities $P_0=(1+\mu_0^2)/2$ and $P_1=(1+\mu_1^2)/2$. Time is in units of inverse $\gamma$. Bottom estimate is $\ln\left(\frac{1-\sigma \mu_0}{1-\sigma \mu_1}\right)$. Upper estimate is plotted as $e+\ln\left(\frac{1-\sigma \mu_0}{1-\sigma \mu_1}\right)$, i.e. neglecting term $\pi/\omega$ since typically $\omega\ll\gamma$.

\begin{figure}
\centering
\includegraphics[width=\linewidth]{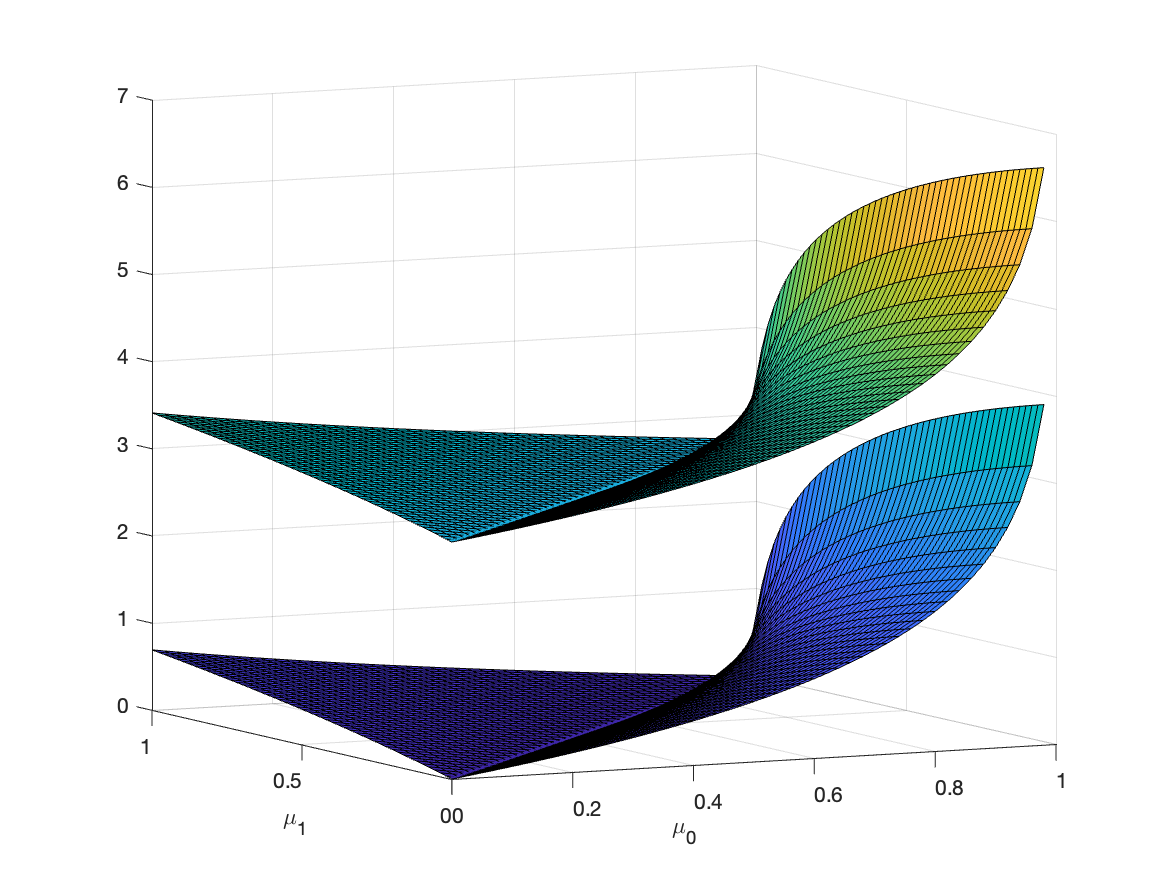}
\caption{This figure shows lower and upper estimates for the minimal time as a function of purity of the initial and final states. The upper estimate is plotted neglecting term $\pi/\omega$ which is typically several orders of magnitude smaller than term $e/\gamma$. Time is in the units of inverse $\gamma$.}
\label{Fig1}
\end{figure}


\section{Controlled qubit interacting with the environment}

State of a qubit is described by a density matrix $\rho$, which is a $(2\times2)$ positive definite Hermitian matrix with unit trace, $\rho\in\mathbb C^{(2\times 2)}$, $\rho\ge 0$, $\tr\rho=1$. The control of the qubit is achieved through coherent interaction in the Hamiltonian. The influence of the environment leads to the appearance of incoherent terms~\cite{PechenRabitz2006}, for which we consider a GKSL form. Let $\sigma_x$, $\sigma_y$, and $\sigma_z$ be the standard Pauli matrices. Then, the free and interaction Hamiltonians of the unitary part of the master equation~\eqref{eq:master} have the form $H=\frac12\omega\sigma_z$, where $\omega>0$ is the qubit's eigentransition frequency, and $V=\frac12\kappa\sigma_x$, where $\kappa>0$ is the coefficient of coupling of the qubit to the external coherent control. Thus, the control system is described by the master equation
\begin{equation}
\label{eq:main_quantum_system}
	\frac{{\rm d\,} \rho(t)}{{\rm d\,} t} =
	-\frac{{\rm i\,}}{\hbar}\Bigl[\frac12\omega \sigma_z + \frac12\kappa u(t)\sigma_x, \rho(t)\Bigr] 
	+\gamma \Big( \sigma^- \rho(t) \sigma^+ 
	- \frac{1}{2} \Big\{ \sigma^+ \sigma^-, \rho(t) \Big\} \Big).
\end{equation}
Here, $u(t)\in\mathbb{R}$ is the coherent control, and the last GKSL term represents the non-unitary interaction with the external zero-temperature environment, where $\gamma\ge 0$ is the decoherence coefficient,
\[
	\sigma^+=\left(\begin{array}{cc}
	0 & 1 \\
	0 & 0  
	\end{array}\right)
	\quad\text{and}\quad
	\sigma^-= \left(\begin{array}{cc}
	0 & 0 \\
	1 & 0
	\end{array}\right)
\]
denote the creation and annihilation operators, $[\cdot,\cdot]$ and $\{\cdot,\cdot\}$ represent the commutator and anticommutator of matrices, respectively. For the chosen qubit implementation, $\omega$, $\kappa$, and $\gamma$ are constants. Furthermore, we set $\hbar=1$.

The paper addresses the time minimization problem $T\to\min$ for equation~\eqref{eq:main_quantum_system} with certain boundary conditions. The optimal motion time from state $\rho^0$ to state $\rho^1$ is defined as the infinum of all times $T$ necessary for the system to transition from the initial state $\rho(0)=\rho^0$ to the final state $\rho(T)=\rho^1$ using all possible control functions $u(t)$. If no control function $u(t)$ is able to change the state from $\rho^0$ to $\rho^1$, we set $T=+\infty$.

In addition to the two-point time minimization problem, we will consider more general formulations of the form $\rho(0)\in M_0$ and $\rho(T)\in M_1$ for some sets $M_0$ and $M_1$. In such formulations, we define the minimum motion time considering all possible choices of control $u(t)$ and initial and final states, as long as they belong to $M_0$ and $M_1$, respectively. In particular, an important problem is the fastest attainment of a given purity level of the final state $P(\rho(T))$ from a given purity level of the initial state $P(\rho(0))$.

\begin{remark}
    Since the system~\eqref{eq:main_quantum_system} is affine with respect to the control $u(t)$, it is natural to assume that $u\in L_1$. However, such a choice of the control class is not suitable for this problem. Since the set of controls is unbounded, i.e., $u\in\R$, the time minimization problem $T\to\min$ may not have (and most likely will not have) an optimal solution in the class $u\in L_1$ as the following example shows:
    \[
        \begin{gathered}
            T\to\min\\
            \rho(0)=\frac12\big(\mathbb I + \frac12\sigma_y\big),\quad 
            \rho(T)=\frac12\big(\mathbb I + \frac12\sigma_z\big).
        \end{gathered}
    \]
    In fact, by rotating using the Hamiltonian $\frac12\kappa u(t)\sigma_x$ with an arbitrary large control $u(t)$, it is possible to connect these two states in an infinitely small time, but $T=0$ cannot be achieved with $u\in L_1$.
    
    Nevertheless, we will show that an optimal control can always be found in the class of impulse controls. For this reason, considering the control in the problem~\eqref{eq:main_quantum_system} as impulsive is much more reasonable. Further details on this question are discussed in Section~\ref{sec:auxiliary_problem}.
\end{remark}

Let us rewrite system~\eqref{eq:main_quantum_system} in terms of Pauli matrices. Density matrix of the qubit can be written in coordinate form as
\[
	\rho=\frac{1}{2}\left(\mathbb I+r_x\sigma_x+r_y\sigma_y+r_z\sigma_z\right),
\]
where $\mathbf{r}=(r_x,r_y,r_z)\in\mathbb{R}^3$. The vector $\mathbf{r}$ belongs to the Bloch ball, $|\mathbf{r}|\leq 1$, since $1-|\mathbf{r}|^2=4\det\rho\geq 0$. Moreover, $1+|\mathbf{r}|^2=2\tr\rho^2$. Thus, $\mu=|\mathbf{r}|$ is naturally correlated with the purity level of the state $\mathbf{r}$. Standard notion of purity is $P(
\rho)=\tr\rho^2=(1+|\mathbf{r}|^2)/2$ and $P(\rho)\in[\frac12;1]$ (since the dimension is 2). So $|\mathbf{r}| = \sqrt{2\tr\rho^2-1}$ and states with a given level of purity $P(\rho)\ne\frac12$ form a sphere in the Bloch ball of radius $\sqrt{2P(\rho)-1}$. Pure states satisfy $|\mathbf{r}|=1$ and belong to the Bloch sphere. Mixed states satisfy $|\mathbf{r}|<1$ and lie strictly inside the Bloch sphere. 

In these coordinates, the control system~\eqref{eq:main_quantum_system} takes the following form:
\begin{equation}
\label{eq:general_control_system}
	\dot{\rr} = \omega f_0(\rr) + \kappa f_1(\rr)u,
\end{equation}
where (see~\cite{MorzhinIJTP2019} for details)
\begin{eqnarray*}
	f_0(\rr)
	&=&
	\left(\begin{array}{ccc}
		0&-1&0\\
		1&0&0\\
		0&0&0
	\end{array}\right)\rr
	-
	\frac{\gamma}{\omega}
	\left(\begin{array}{ccc}
		\frac{1}{2} & 0 & 0\\
		0 & \frac12 & 0\\
		0 & 0 & 1
	\end{array}\right)\rr
	+
	\frac{\gamma}{\omega}
	\left(\begin{array}{ccc}
		0\\0\\1
	\end{array}\right);\\
	f_1(\rr)&=&
	\left(\begin{array}{ccc}
		0&0&0\\
		0&0&-1\\
		0&1&0
	\end{array}\right)\rr.
\end{eqnarray*}
In the absence of external coherent control ($u\equiv 0$), the system state exponentially quickly approaches the north pole $\mathbf{r}=(0,0,1)$ as $t\to\infty$.


\section{Auxiliary 2D control problem in cylindrical coordinates}
\label{sec:auxiliary_problem}

Since the control $u\in\R$ allows for rotation of the system around the $Ox$ axis with arbitrarily large, a priori unbounded, angular velocity, it is natural to switch to cylindrical coordinates $(r_x,R,\theta)$ with the axis aligned along the $Ox$ axis. Set
\[
r_y=R\cos\theta;
\quad r_z=R\sin\theta.
\]
Then the control system~(\ref{eq:general_control_system}) takes the following form
\begin{equation}
\label{eq:main_system_in_R_theta}
    \begin{cases}
        \begin{pmatrix}
            \dot r_x\\
            \dot R
        \end{pmatrix}
        &=
        \begin{pmatrix}
            -\frac12\gamma & -\omega\cos\theta\\
            \omega\cos\theta & -\frac12\gamma(1+\sin^2\theta)
        \end{pmatrix}
        \begin{pmatrix}
            r_x\\
            R
        \end{pmatrix}
        +
        \begin{pmatrix}
            0\\
            \gamma\sin\theta
        \end{pmatrix}\\[0.5cm]
        \quad\dot\theta &=\kappa u -
        \omega\left(
            \frac14\frac\gamma\omega\sin2\theta +
            \frac1R\Big(
                r_x\sin\theta - \frac\gamma\omega\cos\theta
            \Big)
        \right)
    \end{cases}
\end{equation}

\begin{remark}
\label{rm:not_L1_but_impulse}
    In the $(r_x, R, \theta)$ coordinates, it is easy to see why there might not exist an optimal control in the class $L_1$. For instance, let us consider the case where the initial and final states satisfy $r_x^0 = r_x^1$, $R^0 = R^1$, but $\theta^0 \neq \theta^1$. In this scenario, $\inf T = 0$ (excluding some degenerate cases), but this infimum is not achieved by any control $u \in L_1$. However, a control with a single impulse at time zero precisely achieves $T = 0$.

\end{remark}

Extending the class of controls $u(t)$ automatically extends the class of functions for $\theta(t)$. The following change of variables for the control helps to choose a suitable class of functions for $\theta(t)$:
\begin{equation}
\label{eq:u_v_via_theta}
	\begin{gathered}
		v=\frac\kappa\omega u-
			\frac14\frac\gamma\omega\sin2\theta-
			\frac1R\Big(
                r_x\sin\theta-
                \frac\gamma\omega\cos\theta
            \Big);\\
		u=\frac\omega\kappa\left(
			v + 
			\frac14\frac\gamma\omega \sin2\theta +
            \frac1R\Big(
                r_x\sin\theta -
			    \frac\gamma\omega\cos\theta
            \Big)
			\right).
	\end{gathered}
\end{equation}

\begin{remark}
\label{rm:R_zeros}
    It should be noted that this substitution is non-degenerate only when $R(t)\neq 0$. Therefore, zeros of the function $R(t)$ require additional investigation.
\end{remark}

After the substitution \eqref{eq:u_v_via_theta}, the last equation in the control system \eqref{eq:main_system_in_R_theta} takes the form:
\begin{equation}
\label{eq:dot_theta_v}
    \dot\theta = \omega\,v.
\end{equation}

Thus, by choosing a suitable control function $v(t) \in L_1$, the last coordinate $\theta(t)$ can be made arbitrarily close in the $L_1$ metric to any given measurable function from $L_1$ (since the Sobolev space $W^1_1$ is dense in $L_1$). Therefore, it is natural to consider $\theta(t)$ as an arbitrary measurable function, $\theta(\cdot) \in L_1([0;T] \to \mathbb{R}/2\pi)$. Indeed, the first two equations of system~\eqref{eq:main_system_in_R_theta} are linear with respect to $(r_x, R)$, so for any measurable function $\theta(t)$ and any initial data $(r_x(0), R(0))$, these equations have a unique solution.

Hence, the natural class for the functions $\theta(t)$ is the class of measurable functions $L_1$. At first glance, the seemingly strange assumption of the discontinuity of the phase coordinate $\theta(t)$ is actually well justified, as we will show that the optimal function $\theta(t)$ chosen from all measurable functions in $L_1$ must be continuous on the entire sengment $[0;T]$ and analytic on the half-open segment $(0;T]$. So, from the perspective of the original problem, we are seeking the optimal control in a vast class of generalized functions $u(\cdot)\in H^{-1}$. However, it turns out that the optimal control $u(t)$ in this class contains only two impulses at the initial and final time moments $t=0,T$, while on the interval $(0;T)$ it is analytic (see Theorem~\ref{thm:times_original_auxiliary_equal}).

Therefore, instead of the time minimization problem for the original system~\eqref{eq:main_quantum_system}, it is natural to consider the following key auxiliary problem, in which there is no control $u$ (or equivalently, $v$), but instead the variable $\theta$ serves as the new control:
\begin{equation}
\label{eq:auxiliary_system}
	\begin{array}{c}
		T\to\min\\[0.2cm]
		(r_x(0),R(0))\in M_0;\quad (r_x(T),R(T))\in M_1;\\[0.2cm]
		\dfrac{1}{\omega}
		\begin{pmatrix}
			\dot r_x\\
			\dot R 
		\end{pmatrix}
		=
		\begin{pmatrix}
			-\frac12\frac\gamma\omega & -\cos\theta;\\
			\cos\theta  & -\frac12\frac\gamma\omega(1+\sin^2\theta)\\
		\end{pmatrix}
		\begin{pmatrix}
			r_x\\ 
			R
		\end{pmatrix}
		+
		\begin{pmatrix}
			0\\
			\frac\gamma\omega\sin\theta.
		\end{pmatrix}
	\end{array}
\end{equation}
Here, function $\theta(t)$ can be chosen arbitrarily and therefore acts as a control, while the sets $M_0$ and $M_1$ are selected depending on the type of problem and can be single points (in which case we obtain a point-to-point time minimization problem), but not necessarily.

\begin{remark}
\label{rm:2d_3d}
    The auxiliary system~\eqref{eq:auxiliary_system} has two significant advantages over the original system~\eqref{eq:main_quantum_system}. Firstly, the original system has 3D phase space $\R^3=\{\rr=(r_x,r_y,r_z)\}$, while the auxiliary system has 2D phase space $\R^2=\{(r_x,R)\}$. Secondly, 1-dimentional control $u$ in the original system forms a 1-demintional set of admissible velocities at any point (namely, a straight line), while 1-dimentional control $\theta$ in the auxiliary system forms boundary of a 2-dimentional set of admissible velocities. It is quite obvious that controlling a state by a two-dimensional set of admissible velocities in a two-dimensional phase space is much easier than controlling a state by a one-dimensional set of admissible velocities in a three-dimensional phase space.
\end{remark}

\begin{remark}
\label{rm:sing_R}
    It is worth noting that the phase space of the auxiliary problem is a half-disk defined by constraints $r_x^2+R^2\le 1$, $R\ge 0$. This fact is highly inconvenient because both of these constraints are in the phase variables. The constraint $r_x^2+R^2\le 1$ is automatically satisfied since $\frac{\mathrm{d}}{\mathrm{d}t}(r_x^2+R^2)\le 0$ for all points on the circle $r_x^2+R^2=1$, so it can be ignored. However, the constraint $R\ge 0$ is particularly inconvenient, since optimal control problems with phase constraints are generally much harder to analyze comparing to problems without such constraints. To eliminate the constraint $R\ge 0$, we can exploit the discrete symmetry in the auxiliary problem~(\ref{eq:auxiliary_system}):
    \[
    	R\mapsto-R;\quad \theta\mapsto\theta+\pi.
    \]
    Thus, we consider the phase space of the problem~(\ref{eq:auxiliary_system}) as the unit ball $r_x^2+R^2\le 1$, since any trajectory within the ball can be mapped to a trajectory within the semiball $R\ge0$ by applying the reflection $R\mapsto-R$, $\theta\mapsto\theta+\pi$ on those time intervals where $R(t)<0$, that is,
    \[
    	R(t)\mapsto |R(t)|;\quad
    	\theta(t)\mapsto \theta(t) + \frac\pi2(1-\sgn\,R(t))
    \]
    Therefore, the control $\theta(t)$ may have discontinuities at the time moments when $R(t)=0$. Discontinuities of $\theta(t)$ produce impulses in the original control $u(t)$, which is highly inconvenient and arise in some quantum control problems (see e.g.~\cite{Gurman}). However, a nontrivial result is that the optimal trajectory in the open qubit controlling problem~\eqref{eq:main_system_in_R_theta} has no such time moments on $(0;T)$ (see Corollary~\ref{cor:R_zeros}).

\end{remark}

\bigskip

The main convenience of the auxiliary problem~\eqref{eq:auxiliary_system} lies in its two-dimensionality, and the set of controls being essentially compact because the right-hand side of the system~\eqref{eq:auxiliary_system} is periodic with respect to $\theta$, meaning that $\theta\in\R/2\pi\mathbb{Z}$. Moreover, we will show that when the angle $\theta$ traverses the interval $[0;2\pi]$, the right-hand side of the system~\eqref{eq:auxiliary_system} (with fixed $r_x$ and $R$) moves along a closed curve that forms the boundary of a convex compact set. This property makes the problem~\eqref{eq:auxiliary_system} very convenient for investigation using geometric control theory methods.


\section{Obtaining a desired level of purity}

In addition to the point-to-point time minimization problem, an important problem is to obtain a desired level of state purity in minimum time. Obtaining exact solutions to this problem is the first step of the investigation of the original problem~\eqref{eq:master}, \eqref{problem:master}. Moreover, the problem is of great significance in quantum information and quantum computing since the state purity serves as a measure of accuracy and reliability of quantum operations and algorithms.

The problem, in terms of the auxiliary 2D system~\eqref{eq:auxiliary_system}, can be formulated as follows:
\begin{equation}
\label{problem:pureness}
	\begin{array}{c}
		T\to\inf\\[0.2cm]
		r_x^2(0)+R^2(0)=2P_0-1;\quad r_x^2(T)+R^2(T)=2P_1-1,
	\end{array}
\end{equation}
where $P_0$ and $P_1$ are computed from given levels of purity for the initial and final states, respectively. Finding the exact solution to this problem provides a lower bound on the optimal time of evolution in the original problem~\eqref{eq:master}, \eqref{problem:master}. Indeed, the optimal time of evolution from any initial state $\rho^0$ to any final state $\rho^1$ cannot be smaller than the optimal time of evolution in the problem~\eqref{problem:pureness} with the corresponding purity levels $P_0$ and $P_1$.

\begin{proposition}
\label{prop:pureness}

    Put $\sigma=\text{sgn}(P_1-P_0)=\pm1$. If $P_1<1$, then there exists an optimal solution to problem~\eqref{problem:pureness} for system~\eqref{eq:auxiliary_system}. On any optimal solution, $R(t)$ does not change its sign. Without loss of generality, we assume $R(t)\ge 0$, and in this case,
	\[
		r_x(t)\equiv 0,\quad
		R(t) = \mu_0e^{-\gamma t} + \sigma (1-e^{-\gamma t}),\quad
		\theta(t) \equiv \sigma\frac\pi2,
		\quad\text{and}\quad
		T=\frac1\gamma\ln\frac{1-\sigma \mu_0}{1-\sigma \mu_1}.
	\]
    where
    \[
        \mu_{0,1} = \sqrt{2P_{0,1}-1}.
    \]
    If $P_1=1$ and $P_0<1$, then there is no optimal solution because there is no control that can bring the system to the purity level of $P_1=1$ in finite time.
\end{proposition}

The proof is provided in Appendix~\ref{sec:appendix_proof_prop:pureness}.


\section{Continuity and analyticity of the optimal control in the auxiliary problem~(\ref{eq:auxiliary_system})}

\begin{proposition}
\label{prop:auxiliary_existence}

    Let us assume that the terminal sets $M_0$ and $M_1$ within the disk ${r_x^2+R^2\le 1}$ are closed, and at least one point in $M_1$ is reachable from some point in $M_0$ using the auxiliary control system~\eqref{eq:auxiliary_system}. Then there exists an optimal control $\hat\theta(t)\in L_1(0;T)$ in the time minimization problem~\eqref{eq:auxiliary_system}. Moreover, any optimal control $\hat\theta(t)\in L_1(0;T)$ is in fact a continuous function of time on the segment $[0;T]$, $\hat\theta(\cdot)\in C[0;T]$. The corresponding optimal trajectory $(\hat r_x(t),\hat R(t))$ lies in $C^1[0;T]$, and all three functions $\hat r_x(t)$, $\hat R(t)$, and $\hat\theta(t)$ are analytic on the half-segment $(0;T]$. If additionally $|\hat R(0)|<1$, then all three functions are analytic for all $t\in[0;T]$. If at some moment $\tau$ we have $\hat R(\tau)=0$, then $\frac1\omega|\dot R(\tau)|=\sqrt{r_x^2(\tau)+\frac{\gamma^2}{\omega^2}}\ne0$, $\cos\theta(\tau)=\omega r_x(\tau)/\dot R(\tau)$, and $\sin\theta(\tau)=\gamma/\dot R(\tau)$.
\end{proposition}

The proof is provided in Appendix~\ref{sec:appendix_proof_prop:auxiliary_existence}.


\section{Absence of Impulses}

Next, we investigate the set of intersections of the optimal trajectory with the line $R=0$. For this purpose, let us denote
\[
    Z=\{0<\tau<T:R(\tau)=0\}.
\]
The structure of the set $Z$ is of great importance, e.g., since the control $u(t)$ has impulses at $t\in Z$ (see Remarks~\ref{rm:R_zeros} and~\ref{rm:sing_R}). From Proposition~\ref{prop:auxiliary_existence}, it follows, due to the analyticity of $R(t)$, that the number of elements in $Z$ is finite, $\#Z<\infty$. In fact, this result can be significantly improved.

\begin{corollary}
\label{cor:R_zeros}
    On any optimal trajectory in problem~\eqref{eq:auxiliary_system}, it holds that $\#Z\leq 1$. Moreover, if $\#Z=1$ and $Z=\{\tau\}$, then $R(t)<0$ for $t\in[0;\tau)$ and $R(t)>0$ for $t\in(\tau;T]$, or vice versa.
\end{corollary}

Thus, the optimal trajectory for $0<t<T$ intersects the diameter $R=0$ at most once. By using the fact that $R(0)\geq 0$ and $R(T)\geq 0$, it can be shown that such intersections do not occur at all (a detailed explanation is provided in the proof of Theorem~\ref{thm:times_original_auxiliary_equal}).

\begin{proof}[Proof of Corollary~\ref{cor:R_zeros}]

    The essence of the statement is that the second zero of the function $R(t)$ will always be a Maxwell point, after which the trajectory cannot be optimal due to analyticity. Below is a detailed proof. 
    
    Let's consider an arbitrary optimal trajectory $(r_x(t), R(t), \theta(t))$ in the problem~\eqref{eq:auxiliary_system}. Now, let's examine the reflected trajectory $(r_x(t), -R(t), \theta(t) + \pi)$. The reflected trajectory satisfies the system of differential equations in problem~\eqref{eq:auxiliary_system} because the symmetry $R \mapsto -R$, $\theta \mapsto \theta + \pi$ preserves the system.

    Sets $M_0$ and $M_1$ may not be symmetric with respect to the reflection $R \to -R$, so the initial and final points of the reflected trajectory may not lie in $M_0$ and $M_1$, respectively. However, the reflected trajectory is useful in constructing a composite nonsmooth but optimal trajectory when the function $R(t)$ has multiple zeros.

    Note that if $\#Z=1$, $Z=\{\tau\}$, then on each of the intervals $(0;\tau)$ and $(\tau;T)$, the continuous function $R(t)$ has constant but different signs, since $\dot R(\tau)\ne 0$ according to Proposition~\ref{prop:auxiliary_existence}. Therefore, we need to prove that either $\#Z=0$ or $\#Z=1$, but $R(0)\ne 0$ and $R(T)\ne 0$.

    Let us proof the Corollary by contradiction. Suppose that $\#Z\ge 2$ or $\#Z=1$, but $R(0)=0$ or $R(T)=0$. Then there exist $0\le\tau_0<\tau_1\le T$ such that $R(\tau_{0,1})=0$, and at least one of the following conditions holds: $\tau_0>0$ or $\tau_1<T$. Both of these cases are similar, so let's assume that $\tau_1<T$.

    If in problem~\eqref{eq:auxiliary_system} there exist two distinct optimal trajectories that pass through the same point at a particular time $\tau$, then such a time is called a Maxwell point. In the analytic case, a Maxwell point usually contradicts optimality. Indeed, consider a composite trajectory $(r_x(t),\tilde R(t),\tilde \theta(t))$, where
    \[
        \tilde R(t) =
        \begin{cases}
            R(t),&\text{for }t\notin[\tau_0;\tau_1];\\
            -R(t),&\text{for }t\in[\tau_0;\tau_1];
        \end{cases}
        \qquad
        \tilde\theta(t) =
        \begin{cases}
            \theta(t),&\text{for }t\notin[\tau_0;\tau_1];\\
            \theta(t)+\pi,&\text{for }t\in[\tau_0;\tau_1];
        \end{cases}
    \]
    The function $\tilde R(t)$ is Lipschitz continuous since $R(\tau_{0,1})=-R(\tau_{0,1})=0$, and the control $\tilde\theta(t)$ is obviously measurable. Therefore, the composite trajectory $(r_x(t),\tilde R(t),\tilde\theta(t))$ is admissible in problem~\eqref{eq:auxiliary_system}. It is clear that $R(0)=\tilde R(0)$ and $R(T)=\tilde R(T)$. Hence, the composite trajectory starts at a point in the set $M_0$ and ends at a point in the set $M_1$. The time of motion along this trajectory coincides with the time $T$ of motion along the original optimal trajectory, so the composite trajectory is also optimal.

    Since $0 < \tau_1 < T$, we have $\dot{\tilde R}(\tau_1-0)=-\dot R(\tau_1)=-\dot{\tilde R}(\tau_1+0)$. However, Proposition~\ref{prop:auxiliary_existence} states that function $\tilde R(t)$ must be analytic due to optimality. Therefore, $\dot{\tilde R}(\tau_1)=0$, which contradicts the last statement in Proposition~\ref{prop:auxiliary_existence}.
\end{proof}


\section{Reconstruction of the Optimal Control from the Auxiliary Problem}

The optimal solutions to the auxiliary problem \eqref{eq:auxiliary_system} completely determine the optimal solutions to the original problem. This is quite usefull since the investigating problem \eqref{eq:auxiliary_system} is much simpler.

\begin{theorem}
\label{thm:times_original_auxiliary_equal}

    The minimum time of motion from the state $\mathbf{r}^0=(r_x^0,r_y^0,r_z^0)$ to the state $\mathbf{r}^1=(r_x^1,r_y^1,r_z^1)$ for the quantum control system \eqref{eq:main_quantum_system} using impulse control $u(t)$ coincides with the minimum time of motion from the state $(r_x^0,R^0)$ to the state $(r_x^1,R^1)$ in the auxiliary problem \eqref{eq:auxiliary_system}, where $R^j=((r_y^j)^2 + (r_z^j)^2)^{1/2}$ for $j=0,1$. Furthermore, the optimal control $\hat u(t)$ in the problem \eqref{eq:main_quantum_system} exists and can be recovered from the optimal control $\hat\theta(t)$ in the auxiliary problem \eqref{eq:auxiliary_system} using the following formula:	
    \begin{multline}
	\label{eq:u_from_theta}
        \hat u(t) = \frac{1}{\kappa}\Bigg[
            \dot{\hat\theta}(t) +
            \omega\left(
                \frac14\frac\gamma\omega\sin2\hat\theta(t) +
                \frac{1}{\hat R(t)}\Big(\hat r_x(t)\sin\hat\theta(t) -
                \frac\gamma\omega \cos\hat\theta(t)\Big)
            \right)+\\
            +(\hat\theta(0)-\theta_0)\delta_0(t) +
            (\theta_1-\hat\theta(T))\delta_T(t)
        \Bigg],
	\end{multline}

    \noindent where $\delta_\tau(\cdot)$ is the Dirac delta function at the time $\tau$, $\theta_0=\theta(0)$ and $\theta_1=\theta(T)$ are the initial and final states of the angle $\theta$ in the original problem.
\end{theorem}

Thus, the optimal control $u(t)$ is a sum of two impulses (Dirac delta functions) at the initial and final time moments and some function that is analytic on the half-open interval $(0, T]$. At $t=0$, this function may not be analytic, but in that case, it has an integrable singularity as $t\to +0$. Moreover, if $r_y^2(0)+r_z^2(0)\neq 1$, then this function is analytic over the entire interval $[0, T]$.

\begin{proof}[Proof of Theorem~\ref{thm:times_original_auxiliary_equal}]

    It is easy to see that the minimum time of motion in the auxiliary problem \eqref{eq:auxiliary_system} is not greater than the time of motion in the original problem \eqref{eq:main_quantum_system}. Indeed, for any control $u(t)\in L_1$, we determine the sulution $\rr(t)$ of the system \eqref{eq:general_control_system}. Using $\rr(t)$ we can easily find $R(t)=\sqrt{r_y(t)^2 + r_z(t)^2}$, but the angle $\theta(t)$, can not be determined in general. Indeed, we seek for a number $\theta(t)$, such that $r_y(t)=R(t)\cos\theta(t)$ and $r_z(t)=R(t)\sin\theta(t)$. So if $R(t)=0$, then $\theta(t)$ is not unique. Nonetheless, we claim the $R(t)=0$ on a set of zero measure. Indeed, if $S=\{t:R(t)=0\}$ has positive measure, then $r_y(t)=r_z(t)=0$ for $t\in S$ and $\dot r_y(t)=\dot r_z(t)=0$ for a.e.\ $t\in S$ since a.e.\ point in $S$ is a Lebesgue point. Hence for a.e.\ $t\in S$ we have $0=\dot r_z(t)=\gamma/\omega$ that contradicts $\gamma\ne0$. So $S$ has zero measure, and it is closed as $R(t)$ is a continuous function. Hence on an open dence set in $[0;T]$, function $\theta(t)$ is a nice continuous and bounded function. Summarizing, $\theta(t)$ is bounded and maesurable, so $\theta(\cdot)\in L_\infty(0;T)$. Therefore, we can use the obtained function $\theta(t)$ as a control in the auxiliary system \eqref{eq:auxiliary_system}, where we put $M_0=\{(r_x^0,R^0)\}$ and $M_1=\{(r_x^1,R^1)\}$. Consequently, the minimum time of motion in the auxiliary problem cannot be greater than that in the original problem.

    Now let's try to solve the inverse problem and recover the control $u(t)$ in the original problem based on the control $\theta(t)$ in the auxiliary problem. This problem cannot be solved within the class of measurable controls $u(t)$ because if $\theta(t)$ is some measurable function, then $\dot\theta(t)$ is no longer a measurable function but rather a generalized function from the excessively broad class $H^{-1}(0;T)$.

    The situation is saved by the fact that according to Proposition \ref{prop:auxiliary_existence}, in the auxiliary problem, the optimal control $\hat\theta(t)$ exists as long as the initial and final points can be connected by at least one admissible curve. Therefore, it is sufficient to recover the control $\hat u(t)$ only from the optimal control $\hat\theta(t)$ in the auxiliary problem, which can be done since the optimal control $\hat\theta(t)$ must be continuous on $[0,T]$ and analytic on $(0,T]$ according to Proposition \ref{prop:auxiliary_existence}.

    So, let $\hat\theta(t)$ be an optimal control in the auxiliary problem \eqref{eq:auxiliary_system}. To construct the control $\hat u(t)$ from $\hat\theta(t)$, we need to overcome two difficulties:
    
	\begin{itemize}
		
		\item In general, when transitioning to the auxiliary problem, the information about the initial and final angles $\theta^0$ and $\theta^1$ is lost. Therefore, it is possible that the optimal control $\hat\theta(t)$ (for the auxiliary problem \eqref{eq:auxiliary_system}) does not satisfy the equalities $\hat\theta(0)=\theta^0$ and $\hat\theta(T)=\theta^1$. To address this issue, we introduce impulses at the initial and final time moments, which transform $\theta^0\mapsto\hat\theta(0)$ and $\hat\theta(T)\mapsto\theta^1$.
		
		\item Since $R^0$ and $R^1$ have the same sign (both numbers are non-negative), on the optimal trajectories of the auxiliary problem \eqref{eq:auxiliary_system}, the variable $\hat R(t)$ cannot become zero for $t\in(0;T)$. Indeed, if $R^0>0$ and $R^1>0$, the number of zeros of the function $\hat R(t)$ must be even (since the derivative $\dot{\hat R}$ does not vanish at such points, according to Proposition \ref{prop:auxiliary_existence}), hence $\#Z=0$ by Corollary \ref{cor:R_zeros}, and so $R(t)>0$. If either $R^0=0$ or $R^1=0$, again we have $\#Z=0$ according to Corollary \ref{cor:R_zeros}, and again $\hat R(t)>0$. If $R^0=R^1=0$, then $\#Z=0$, but it may happen that $R(t)<0$ for $0<t<T$ in this case.  Nonetheless, the reflection $R\mapsto -R$, $\theta\to\theta+\pi$ does not change formula \eqref{eq:u_from_theta}. Therefore, even if $R^0=R^1=0$ and $\hat R(t)<0$ for $t\in(0;T)$, we can consider $\tilde R(t)=-\hat R(t)>0$  and $\tilde\theta(t)=\hat\theta(t)+\pi$ instead of $\hat R(t)$ and $\hat\theta(t)$. Thereby we assume that $\hat R(t)>0$ without loss of generality.
  
  	\end{itemize}	

    So, the control $\hat v(t)$ according to~\eqref{eq:dot_theta_v} is given by the following formula:
	\[
		\hat v(t) = \frac1\omega\left(
			\dot{\hat\theta}(t) + 
			(\hat\theta(0)-\theta_0)\delta_0(t) +
			(\theta_1-\hat\theta(T))\delta_T(t)
		\right).
	\]
    According to Proposition \ref{prop:auxiliary_existence}, function $\dot{\hat\theta}(t)$ is analytic on $(0;T]$. Moreover, if $R^0\neq 1$, then function $\dot\theta(t)$ is analytic on $[0;T]$, and if $R^0=1$, then function $\dot{\hat\theta}(t)$ may not be analytic at $t=0$, but it has an integrable singularity at that point since function $\hat\theta(t)$ is continuous at $t=0$.

    The original optimal control $\hat u(t)$ is determined by formula~\eqref{eq:u_v_via_theta}. The corresponding term contains a factor involving $\hat R(t)$ in the denominator, thus its analyticity at the points $\tau=0,T$ needs to be further investigated if $R(\tau)=0$. Let us assume $R(0)=0$ (the case $R(T)=0$ is similar). It is necessary to investigate the behavior of the function
    \[
        g(t)=\frac1{\hat R}\left(\hat r_x\sin\hat\theta - \frac\gamma\omega \cos\hat\theta\right)
    \]
    in a vicinity of $t=0$. Since $\hat R(0)\ne \pm1$, the function $R(t)$ is analytic according to Proposition~\ref{prop:auxiliary_existence}. By the same proposition, we have $\dot{\hat R}(0)\ne 0$, $\cos\hat\theta(0) = \omega \hat r_x(0)/\dot{\hat R}(0)$, and $\sin\hat\theta(0)=\gamma/\dot{\hat R}(0)$. Therefore, the function $g(t)$ has a removable singularity at $t=0$ and is analytic at $t=0$ by the Riemann theorem.

    Thus, the minimum time of motion in the original qubit control problem~\eqref{eq:main_quantum_system} is not less than the minimum time of motion in the auxiliary problem~\eqref{eq:auxiliary_system}, and therefore, they coincide.
\end{proof}


\section{Estimation of the minimum motion time in the state-to-state problem}

\begin{theorem}
\label{thm:time_estimates}
    Let's denote by $T$ the minimum time of motion for the two-point problem $\rho(0)=\rho^0$, $\rho(T)=\rho^1$ for the qubit control system~\eqref{eq:main_quantum_system}, and by $\rr^j=(r_x^j,r_y^j,r_z^j)$, $j=0,1$ the coordinates of the initial and final states $\rho^0$ and $\rho^1$ in the Bloch sphere. If $|\rr^1|<1$, then
    \[
		T\ge \frac1\gamma\ln\frac{1-\sigma |\rr^0|}{1-\sigma |\rr^1|},
	\]
	where $\sigma=\sgn\,(|\rr^1|-|\rr^0|)=\pm1$. If additionall we have $|\rr^1|\le 1- \frac\pi2\frac\gamma\omega$, then
	\[
		T\le \frac{\pi}{\omega} + \frac{e}\gamma +
		 \frac1\gamma\ln\frac{1-\sigma |\rr^0|}{1-\sigma |\rr^1|},
	\]
    However, if $|\rr^0|<1$ and $|\rr^1|=1$, then there is no optimal solution because there exists no control that can drive the system to the purity level $|\rr^1|=1$ in finite time.

\end{theorem}

The proof of Theorem \ref{thm:time_estimates} provides explicit formulas for the control that achieves the upper bound on~$T$ (see~\eqref{eq:theta_suboptimal}, \eqref{eq:v_suboptimal}, and~\eqref{eq:u_from_theta}). This control contains four impulses.

\begin{proof}[Proof of Theorem~\ref{thm:time_estimates}]

    The estimates are based on the investigation of the minimal time of motion in the auxiliary system~\eqref{eq:auxiliary_system} and Theorem~\ref{thm:times_original_auxiliary_equal}.

    The lower estimate for time $T$ is related to the difficulty of changing the purity level of the state $\rr$ in system~\eqref{eq:general_control_system}. Let us assume that some trajectory of the control system transforms the state $\rr^0$ into the state $\rr^1$ in a certain time $T$. Then in the problem of reaching a given purity level $P_1=(1+|\rr^1|^2)/2$ from an arbitrary point at the purity level $P_0=(1+|\rr^0|^2)/2$, the minimal time $\hat T$ cannot be worse than $T$, i.e., $\hat T \leq T$. Indeed, it is possible to achieve the desired purity level within the time $T$ simply by following the original trajectory. Therefore, in terms of the auxiliary system~\eqref{eq:auxiliary_system}, we obtain the problem~\eqref{problem:pureness}. The claimed lower estimate for the time $T$ immediately follows from Proposition~\ref{prop:pureness}.

    The upper estimate is obtained as follows: it is sufficient to construct a trajectory that realizes the given time of motion. Doing this for the original 3D control system~\eqref{eq:main_quantum_system} by 1D control can be quite difficult. However, according to Theorem~\ref{thm:times_original_auxiliary_equal}, it is enough to provide such a trajectory for the auxiliary 2D control system~\eqref{eq:auxiliary_system}, which has in fact 2D set of admissible velocities (see Remark~\ref{rm:2d_3d}).

    The initial part of the motion is designed as follows: from the initial state $(r^0_x, R^0)$ within a time interval
	\[
		\tau_0 = \frac1\omega \left|\arctan \frac{r_x^0}{R^0}\right|
	\]
	by using the control
	\[
		\theta=\frac\pi2\left(
            1-\sgn\arctan \frac{r_x^0}{R^0}
        \right)
	\]
	we reach a state with $r_x(\tau_0)=0$. At that state we have
	\[
		R(\tau_0)=\tilde R^0:=|\rr^0|e^{-\frac\gamma2\tau_0}.
	\]

    The middle part is a motion along the line $r_x(t)\equiv0$, and this motion is described latter.

    The final part of the motion is arranged similarly to the initial part: to complete the motion at the point $(r^1_x, R^1)$ at a certain time $\tilde T$, it is necessary to move from the point with $r_x=0$ for a duration of
	\[
		\tau_1 = \frac1\omega \left|\arctan \frac{r_x^1}{R^1}\right|
	\]
	by using the control
	\[
		\theta=\frac\pi2\left(
            1+\sgn\arctan \frac{r_x^1}{R^1}
        \right)
	\]
	Then, if at time moment $\tilde T - \tau_1$ we have $r_x(\tilde T-\tau_1)=0$ and 
	\[
		R(\tilde T - \tau_1)=\tilde R^1:=|\rr^1|e^{\frac\gamma2\tau_1}
	\]
    then the motion at time $\tilde T$ will indeed end at the desired point $(r_x^1, R^1)$. Note that the point $(r_x, R) = (0, \tilde R^1)$ is reachable from the state $(r_x, R) = (0, \tilde R^0)$ since (according to the conditions of the theorem) $|\rr^1|\le 1-\frac\pi2\frac\gamma\omega$, and therefore,
 	\[
		\tilde R^1 \le 
		\big(1-\frac\pi2\frac\gamma\omega\big)e^{\frac\pi4\frac\gamma\omega}\le
		e^{-\frac\pi2\frac\gamma\omega}e^{\frac\pi4\frac\gamma\omega}=
		e^{-\frac\pi4\frac\gamma\omega} \le 1.
	\]
    Moreover, $\frac{\gamma}{\omega}\leq \frac{2}{\pi}$, since $|\rr^1|\geq 0$.

    It remains to organize the movement in the intermediate time interval $t\in[\tau_0;\tilde T-\tau_1]$ from the point $(r_x,R)=(0,\tilde R^0)$ to the point $(r_x,R)=(0,\tilde R_1)$. This movement takes the longest time, but it can be easily organized optimally using Proposition~\ref{prop:pureness}. Specifically, it is necessary to use the control
    \[
		\theta=\tilde\sigma\frac\pi2
	\]
    during the time
	\[
		\tau_{\frac12} = \frac1\gamma\ln\left(\frac{1-\tilde\sigma\tilde R^0}{1-\tilde\sigma\tilde R^1}\right)
	\]
	where $\tilde\sigma=\sgn(\tilde R^1-\tilde R^0)$.
	
	\medskip

    Therefore, the final state $(r_x^1,R^1)$ can be reached from the initial state $(r_x^0,R_x^1)$ within a time interval of
    \[
		\tilde T = \tau_0 + \tau_{\frac12} + \tau_1
	\]
	by using the folling control
	\begin{equation}
	\label{eq:theta_suboptimal}
		\theta(t) = \begin{cases}
			\frac\pi2(1-\sgn\arctan \frac{r_x^0}{R^0}),&\text{for } t\in(0;\tau_0)\\
			\frac\pi2\sgn(\tilde R^1-\tilde R^0),&\text{for } t\in(\tau_0;\tau_0+\tau_{\frac12})\\
			\frac\pi2(1+\sgn\arctan \frac{r_x^1}{R^1}),&\text{for } t\in(\tau_0+\tau_{\frac12},\tilde T)\\
		\end{cases}
	\end{equation}
    The control $v$ is determined by the formula $v=\frac1\omega\dot\theta$ and contains 4 impulses or fewer:
	\begin{multline}
	\label{eq:v_suboptimal}
		v=\frac1\omega\Big(
			(\theta(0)-\theta^0)\delta_0(t) + 
			\frac\pi2 (-1+\sgn(\tilde R^1-\tilde R^0) + \sgn\arctan \frac{r_x^0}{R^0}) \delta_{\tau_0}(t) +\\
			+\frac\pi2 (1-\sgn(\tilde R^1-\tilde R^0) + \sgn\arctan \frac{r_x^1}{R^1}) \delta_{\tilde T-\tau_1}(t) + 
			(\theta^1-\theta(\tilde T))\delta_0(t)
		\Big).
	\end{multline}
	The control $u$ is obtained using the formula~\eqref{eq:u_from_theta} and, consequently, also has no more than 4 impulses.

    Since the optimal travel time $T$ does not exceed the proposed $\tilde T$, $T \leq \tilde T$, we only need to estimate $\tau_0$, $\tau_{\frac12}$, and $\tau_1$. It is obvious that
	\[
		\tau_0+\tau_1 \le \frac\pi2 \frac1\omega + \frac\pi2 \frac1\omega =\frac\pi\omega.
	\]

	It remains to find a convenient upper bound for the time $\tau_{\frac12}$. The first step is as follows:
	\[
	\ln\left(
		\frac{1-\tilde\sigma\tilde R^0}
			{1-\tilde\sigma\tilde R^1}
	\right) -
	\ln\left(
		\frac{1-\sigma|\rr^0|}
			{1-\sigma|\rr^1|}
	\right) =
	\ln\left(
		1 + \frac{ |\rr^0|(\sigma - 
			\tilde\sigma e^{-\frac\gamma2\tau_0}) }
		{1-\sigma|\rr^0|}
	\right) +
	\ln\left(
		1 + \frac{ |\rr^1|(\tilde\sigma e^{\frac\gamma2\tau_1} - 
			\sigma) }
		{1-\tilde\sigma|\rr^1|e^{\frac\gamma2\tau_1}}
	\right) = a+b.
	\]
 
	Let us consider all 4 possible cases of signs $\sigma=\pm1$ and $\tilde\sigma=\pm1$.
	\begin{enumerate}
		
		\item Let $\sigma=\tilde\sigma=1$. Then $|\rr^0|\le |\rr^1|\le 1-\frac\pi2\frac\gamma\omega\le e^{-\frac\pi2\frac\gamma\omega}$. Since $\frac\gamma\omega\le\frac2\pi$, we obtain
		\[
			a \le \ln\left( 1 +
				\frac{1-e^{-\frac\pi4\frac\gamma\omega}}{1-e^{-\frac\pi2\frac\gamma\omega}} 
			\right) =
			\ln\left( 1 +
				\frac{1}{1+e^{-\frac\pi4\frac\gamma\omega}} 
			\right) \le
			\ln\left(1+\frac1{1+e^{-\frac12}}\right) \le \frac12
		\]
		\[
			b\le \ln\left( 1 +
				\frac{e^{\frac\pi4\frac\gamma\omega}-1}
					{1-(1-\frac\pi2\frac\gamma\omega)e^{\frac\pi4\frac\gamma\omega}} 
			\right) \le
			\ln\left( 1 +
				\frac{e^{\frac\pi4\frac\gamma\omega}-1}{1-e^{-\frac\pi4\frac\gamma\omega}} 
			\right) = 
			\ln (1 + e^{\frac\pi4\frac\gamma\omega}) \le
			\ln (1 + e^{\frac12} ) \le 1
		\]
		Thus, in this case, $a+b\le \frac32$.
		
		\item Let $\sigma=\tilde\sigma=-1$. Then $a+b\le 0$ since $a\le 0$ and $b\le 0$.
		
		\item Let $\sigma=1$ and $\tilde\sigma=-1$. Then $|\rr^1|\ge |\rr^0|$, but $\tilde R^0\ge \tilde R^1$. The last inequality is equivalent to the inequality $e^{-\frac\gamma2\tau_0}|\rr^0|\ge e^{\frac\gamma2\tau_1}|\rr^1|$, which does not contradict to the inequality $|\rr^1|\ge|\rr^0|$ only if $\tau_0=\tau_1=0$. In this case, $\tilde R_0=\tilde R_1=|\rr^0|=|\rr^1|$, and $a+b=0$.
		
		\item Let $\sigma=-1$ and $\tilde\sigma=1$. Then $|\rr^1|\le |\rr^0|$, but $\tilde R^0\le \tilde R^1$. Hence
		\[
			\tilde R^0 \le \tilde R^1 \le e^{\frac\pi2\frac\gamma\omega}\tilde R^0.
		\]
		moreover, $|\rr^1|\le 1-\frac\pi2\frac\gamma\omega\le e^{-\frac\pi2\frac\gamma\omega}$ implies $\tilde R^0\le\tilde R^1\le e^{-\frac\pi4\frac\gamma\omega}$. Therefore
		\[
			\ln\left(
				\frac{1-\tilde R^0}{1-\tilde R^1}
			\right) =
			\ln\left(
				1 + \frac{\tilde R^1-\tilde R^0}{1-\tilde R^1}
			\right)	\le
			\frac{\tilde R^0(e^{\frac\pi2\frac\gamma\omega}-1)}{1-\tilde R^1} \le
			\frac{e^{\frac\pi2\frac\gamma\omega}-1}{e^{\frac\pi4\frac\gamma\omega}-1} =
			1+e^{\frac\pi4\frac\gamma\omega}\le e,
		\]
		as $\frac\gamma\omega\le\frac2\pi$. Additionally, using $|\rr^0|\ge |\rr^1|$ we obtain
		\[
			\ln\left(
				\frac{1+|\rr^0|}{1+|\rr^1|}
			\right) \ge 0.
		\]
		So, in this case, we have $a+b\le e$.
	\end{enumerate}

    Putting together all 4 cases, we obtain that $a+b\le\max\{\frac32,0,0,e\}=e$, which is what was required.
\end{proof}


\appendix

\section{Proof of Proposition~\ref{prop:auxiliary_existence}}
\label{sec:appendix_proof_prop:auxiliary_existence}

To prove the existence of an optimal solution to the auxiliary problem \eqref{eq:auxiliary_system}, we will use Filippov's theorem (see~\cite{AgrachevSachkov}). First, note that the right-hand side of the system satisfies the estimate $|(\dot r_x, \dot R)| \leq c |(r_x, R)|$, implying that all the vector fields involved in the problem are complete. The sets $M_0$ and $M_1$ are compact. Furthermore, the right-hand side of the auxiliary system \eqref{eq:auxiliary_system} for $\theta \in [0, 2\pi]$ forms the boundary of a compact set $U(r_x, R)$. It is evident that the boundary $\partial U(r_x, R)$ is also a compact set and therefore continuously (in the Hausdorff metric) depends on $r_x$ and $R$. Thus, only one condition of Filippov's theorem on the existence of an optimal control is not satisfied. Namely, the boundary of the set $U(r_x, R)$ is not a convex set. However, we can relax system \eqref{eq:auxiliary_system} and consider instead of the original control system $(\dot r_x, \dot R) \in \partial U(r_x, R)$ the convexified system $(\dot r_x, \dot R) \in \mathrm{conv}\, \partial U(r_x, R)$, for which all the conditions of Filippov's theorem are satisfied, and hence an optimal solution exists in the time minimization problem.

Let us demonstrate that the optimal control in the convexified problem is indeed admissible for the original system \eqref{eq:auxiliary_system} (and hence optimal in it as well). To do so, we will show that the set $U(r_x,R)$ is convex for any $r_x$ and $R$. Thus, for any fixed $r_x$ and $R$, the tangent vector $(\xi, \eta)$ to the boundary of the admissible velocity set in the right-hand side of system \eqref{eq:auxiliary_system} at the point $\theta$ has the following form:
\[
	\xi=\frac1\omega\frac{\D \dot r_x}{{\D}\theta} = R\sin\theta;
	\qquad
	\eta=\frac1\omega\frac{\D \dot R}{{\D}\theta} = \frac\gamma\omega\cos\theta\,(1-R\sin\theta)- r_x\sin\theta.
\]
Let us fix $r_x$ and $R$ and show that the vector $(\xi, \eta)$ rotates counterclockwise as $\theta$ increases if $R > 0$, and clockwise if $R < 0$. Indeed,
\begin{equation}
\label{eq:U_curvature}
	\xi'_\theta\eta - \xi\eta'_\theta = \frac\gamma\omega  R(1-R\sin^3\theta).
\end{equation}
Therefore, if $R > 0$, we have $\xi'_\theta\eta - \xi\eta'_\theta\ge 0$ (since $|R|\le 1$). Similarly, if $R < 0$. In the case where $R=0$, we have $\xi'\theta_\eta - \xi\eta'_\theta\equiv 0$, and the set $U(r_x,R)$ becomes a line segment (which is convex).

Thus, the original system \eqref{eq:auxiliary_system} takes the form $(\dot r_x,\dot R)\in\partial U(r_x,R)$, while the convexified system is given by $(\dot r_x,\dot R)\in U(r_x,R)$, since $\mathrm{conv}\,\partial U(r_x,R)=\mathrm{conv}\,U(r_x,R)=U(r_x,R)$ due to the convexity and compactness of $U(r_x,R)$. Let us demonstrate that the optimal control in the new problem with the convexified velocity set always lies on the boundary of the set. To do so, we will apply the Pontryagin's maximum principle to the new problem\footnote{The possibility of applying the Pontryagin's maximum principle is guaranteed here by the fact that the set $U(r_x,R)$ can be parameterized by a two-dimensional control from the unit ball.}. The Pontryagin's function takes the form:
\[
	\mathcal{H} = p\dot r_x + q\dot R
\]
where $(p,q)$ are the conjugate variables to $(r_x,R)$, and $(p,q)\neq(0,0)$.

According to the maximum principle, the optimal control at a.e.\ moment of time maximizes the Pontryagin function $\mathcal{H}$.
\[
	\mathcal{H}\to\max
	\quad\text{w.r.t.}\quad
	(\dot r_x,\dot R)\in U(r_x,R).
\]
Therefore, the optimal velocity $(\dot r_x,\dot R)$ lies on the boundary $\partial U(r_x,R)$, which is determined by the original control $\theta$, as required.

Summarizing, we have proved existence of an optimal solution to problem~\eqref{eq:auxiliary_system}.

\medskip

Let us now investigate the continuity and analyticity of the optimal control. From the Pontryagin's maximum principle, it follows that the optimal control maximizes $\mathcal{H}=p\dot r_x + q\dot R$. When $R\neq 0$, set $U(r_x,R)$ is strictly convex, and when $R=0$, it becomes a vertical line segment. Therefore, if $R\neq 0$, the function $\mathcal{H}$ attains a unique global maximum (up to the period in $\theta$), which is also the global maximum. If $R=0$ and covector $(p,q)$ is not horizontal (i.e., $q\neq 0$), then the global maximum of $\mathcal{H}$ is also unique. Hence, if $R^2+q^2\neq 0$, the function $\mathcal{H}$ has a unique point of maximum, which we denote as $\theta_M(r_x,R,p,q)$.

\begin{lemma}
    If $R^2+q^2\ne 0$, then the function $\theta_M(r_x,R,p,q)$ is continuous.
\end{lemma}

\begin{proof}

    The continuity of the function follows immediately from the uniqueness of the point of global maximum for a continuous function $\mathcal{H}$ on the compact set $\theta\in\R/2\pi\mathbb{Z}$ that continuously depend on parameters $(r_x,R,p,q)$.
\end{proof}

\begin{lemma}
    If $|R|<1$ and $R^2+q^2\ne 0$, then the funvtion $\theta_M(r_x,R,p,q)$ is analitic.
\end{lemma}

\begin{proof}

    At the point of maximum, we have $0=\frac1\omega\mathcal{H}\theta=p\xi+q\eta$. Since the solution $\theta_M(r_x,R,p,q)$ is continuous, to prove its analyticity, it is sufficient to verify that the conditions of the implicit function theorem are satisfied, i.e., $0\ne\frac1\omega\mathcal{H}_{\theta\theta} = p\xi_\theta + q\eta_\theta$. If the implicit function theorem conditions are not satisfied at the point of maximum, then we have the pair of equalities $\mathcal{H}_\theta=\mathcal{H}_{\theta\theta}=0$. The covector $(p,q)$ cannot be zero according to the maximum principle, hence conditions $\mathcal{H}_\theta=\mathcal{H}_{\theta\theta}=0$ imply $\xi_\theta\eta - \xi\eta_\theta=0$. According to the lemma's assumption, $|R|<1$, so from equation~\eqref{eq:U_curvature}, it follows that the conditions of the implicit function theorem can only be violated if $R=0$.

    Let us show that if $R=0$, but $q\ne 0$, then the conditions of the implicit function theorem are satisfied. Indeed, if $R=0$, Then $\frac1\omega\mathcal{H}_\theta=q(\frac\gamma\omega\cos\theta - r_x\sin\theta)$ and $\frac1\omega\mathcal{H}_{\theta\theta}=-q(\frac\gamma\omega\sin\theta + r_x\cos\theta)$. Hence, $\frac1\omega(\mathcal{H}_\theta\cos\theta-\mathcal{H}_{\theta\theta}\sin\theta)=\frac\gamma\omega q\ne 0$.
\end{proof}

Thus, the analyticity of $\theta_M$ can be lost when $|R|=1$, while continuity (and analyticity) can be lost when $R=q=0$. The following two lemmas demonstrate that on the optimal trajectory, the first case can only occur at the initial time instant $t=0$, and the second case can never occur.

\begin{lemma}

    Let $(r_x(t), R(t))$ be a trajectory of the control system with a certain control $\theta(t)$ (not necessarily optimal). If at some time instant $t_0$ we have $r_x^2(t_0) + R^2(t_0) < 1$, then the same inequality holds for all $t \geq t_0$.

\end{lemma}

\begin{proof}

    Indeed,
    \begin{equation}
    \label{eq:pureness_derivative}
	   \frac{1}{\omega}\frac{\D}{\D t}(r_x^2+R^2) = \frac\gamma\omega (1-r_x^2-R^2-(1-R\sin\theta)^2),
    \end{equation}
    Therefore, for any choice of $\theta(t)$, the rate of change of the squared distance from the point $(r_x(t), R(t))$ to the circle of pure states ${r_x^2 + R^2 = 1}$ satisfies the following differential inequality
    \[
	   \frac1\omega\frac{\D}{\D t}(1-r_x^2-R^2) \ge -\frac\gamma\omega (1-r_x^2-R^2),
    \]
    that is, it cannot approach zero faster than an exponential function $ce^{-\gamma t}$. Therefore, if $r_x^2(t_0) + R^2(t_0) < 1$, then $r_x^2(t) + R^2(t) < 1$ for all $t \geq t_0$.
    
\end{proof}

Therefore, if $|R(t_0)|=1$ holds on an optimal trajectory, then $t_0=0$. Indeed, if $|R(t_0)|=1$ for some $t_0>0$, then, according to the lemma, we have $r_x^2(t)+R^2(t)=1$ for all $t\in[0;t_0]$. Moreover, formula~\eqref{eq:pureness_derivative} guarantees that $\theta=\frac\pi2\sgn R(t)$ and $|R(t)|=1$ for almost all $t\in[0;t_0]$. Since the function $R(t)$ is continuous, we obtain $R(t)\equiv R(t_0)=\pm1$ and $r_x(t)\equiv 0$ for all $t\in[0;t_0]$. Thus, if $t_0>0$, the trajectory remains stationary for some time and, therefore, it is not optimal (it does not minimize the motion time).

\begin{lemma}
    On any solution of the Pontryagin's maximum principle, we have $R^2(t)+q^2(t)\neq 0$ for all $t$.
\end{lemma}

\begin{proof}

    Denote
    \[
        \mathbb{T} = \{t:R(t)=q(t)=0\}.
    \]
    We want to show that $\mathbb{T}=\emptyset$. 

    First, we will show that the Lebesgue measure of set $\mathbb{T}$ is zero. We will prove this by contradiction: suppose that the measure of $\mathbb{T}$ is positive. Then almost every point in $\mathbb{T}$ is a Lebesgue point and, therefore, a limit point of $\mathbb{T}$. Since the derivatives $\dot R(t)$ and $\dot q(t)$ exist for almost all $t$, we obtain $\dot R(t) = \dot q(t) = 0$ for almost all $t \in \mathbb{T}$. Thus, for almost all $t \in \mathbb{T}$, we have
    \[
        \begin{aligned}
            0=\dot R =&\,r_x\cos\theta - \frac12 \frac\gamma\omega R(1+\sin^2\theta) + \frac\gamma\omega\sin\theta = r_x\cos\theta + \frac\gamma\omega\sin\theta\\
            0=\dot q =&\,p\cos\theta + \frac12\frac\gamma\omega q(1+\sin^2\theta) = p\cos\theta.
        \end{aligned}
    \]
    The conjugate variables $(p,q)$ cannot both be zero simultaneously. Since $q(t) = 0$ for $t \in \mathbb{T}$, we have $p(t) \neq 0$ for $t \in \mathbb{T}$. Therefore, for almost all $t \in \mathbb{T}$, we have $0 = \cos\theta = r_x\cos\theta + \frac{\gamma}{\omega}\sin\theta$, which is impossible since $\frac{\gamma}{\omega} \neq 0$. Thus, the measure of the set $\mathbb{T}$ is zero.

    Now we will show (again by contradiction) that the set $\mathbb{T}$ is actually empty. If this is not the case, then there exists an instant $t_0$ such that $R(t_0) = q(t_0) = 0$. Since $\mathbb{T}$ has measure zero, there exists an instant $\tau \notin \mathbb{T}$ in any given neighborhood of $t_0$. Without loss of generality, we assume $\tau > t_0$ (the case $\tau < t_0$ is similar). Let $t_1 = \sup{\mathbb{T} \cap (-\infty;\tau)}$. Since $\mathbb{T}$ is closed, we have $t_1 \in \mathbb{T}$ and $t_1<\tau$. Moreover, $(t_1,\tau] \cap \mathbb{T} = \emptyset$. Thus, we have found a moment in time $t_1$ such that $R(t_1) = q(t_1) = 0$ and $R^2(t) + q^2(t) \neq 0$ in the right neighborhood of $t_1$. The PMP system is autonomous, so we can assume that $t_1 = 0$ without loss of generality.

    Let us perform a blow-up of the point $R=q=0$: in polar coordinates, $R=\lambda\cos\varphi$ and $q=\lambda\sin\varphi$, where $\lambda > 0$. We have
    \[
        \begin{cases}
            \frac1\omega\dot\lambda = (r_x\cos\varphi+p\sin\varphi)\cos\theta + \frac\gamma\omega\cos\varphi\sin\theta + O(\lambda)\\
            \frac1\omega\dot\varphi = \frac1\lambda\left[
                (p\cos\varphi - r_x\sin\varphi)\cos\theta - \frac\gamma\omega\sin\varphi\sin\theta + O(\lambda)
            \right]
        \end{cases}
    \]
    The function $\lambda(t)$ is Lipschitz continuous. The function $\varphi(t)$ is locally Lipschitz continuous in a right neighborhood of $t=0$, as in this neighborhood $\lambda(t)>0$. Hence we are not assuming $\varphi\in[0;2\pi]$ to preserve the continuity. Instead, we assume $\varphi\in\R$ which will be very convenient it what follows.

    The Pontryagin function in the new coordinates takes the form
    \[
        \frac1\omega\mathcal{H} = -\frac12\frac\gamma\omega pr_x + \lambda \left[
            (-p\cos\varphi+r_x\sin\varphi)\cos\theta + \frac\gamma\omega\sin\varphi\sin\theta + \alpha
        \right] \to\max_\theta,
    \]
    where the remainder term $\alpha$ is small together with its derivatives with respect to $\lambda$, i.e., $\alpha=O(\lambda)$ and $\alpha'=O(\lambda)$. The key term $S(\theta)=(-p\cos\varphi+r_x\sin\varphi)\cos\theta+\frac\gamma\omega\sin\varphi\sin\theta$ has a unique local (and global) maximum with respect to $\theta$ (denoted as $\theta_S$), satisfying
    \[
        \begin{pmatrix}
            \cos\theta_S\\\sin\theta_S
        \end{pmatrix}=
        \frac1{S(\theta_S)}
        \begin{pmatrix}
            -p\cos\varphi+r_x\sin\varphi\\
            \frac\gamma\omega\sin\varphi
        \end{pmatrix},
        \quad\text{where}\quad
        S(\theta_S)=\sqrt{(r_x\sin\varphi-p\cos\varphi)^2 + \frac{\gamma^2}{\omega^2}\sin^2\varphi}.
    \]
    Moreover, at the maximum point, it is obvious that $S''_{\theta\theta}(\theta_S)=-S(\theta_S)<0$. Therefore, by the implicit function theorem, for small enough $\lambda>0$, we have
    \[
        \theta_M = \theta_S + O(\lambda).
    \]
    Substituting the obtained expansion for $\theta_M$ into the formulas for $\dot\lambda$ and $\dot\varphi$, we obtain
    \[
        \begin{cases}
            \frac1\omega\dot\lambda = \frac1{S(\theta_S)}\left[ -pr_x\cos2\varphi + \frac12(r_x^2 + \frac{\gamma^2}{\omega^2}- p^2)\sin2\varphi\right] + O(\lambda)\\
            \frac1\omega\dot\varphi = -\frac1\lambda\left[
                S(\theta_S) + O(\lambda)
            \right]
        \end{cases},
    \]

    Let's estimate the factors on the right-hand side. On one hand, the function $S(\theta_S)$ is continuous and $S(\theta_S)>0$ if $p\neq 0$, which means that there exists a constant $c$ such that $c\geq S(\theta_S)\geq \frac{1}{c}>0$ for small $t$. On the other hand, the function $\lambda(t)$ is Lipschitz continuous, $\lambda(t)\ge 0$, and $\lambda(t)=0$. Therefore, we have $0\le \lambda(t)\leq ct$ (we can appropriately increase constant $c$).

    Function $S(\theta_S)$ is bounded, positive, and separated from zero in a neighborhood of $t=0$. Therefore, the inequalities $0\leq \lambda(t)\leq ct$ imply $\dot\varphi\leq -\frac{1}{ct}$ (by increasing $c$). Hence, function $\varphi(t)$ monotonically decreases, and $\varphi(t)\to +\infty$ as $t\to +0$, since $\frac{1}{t}$ has non integrable singularity at $0$.

    Since function $\varphi(t)$ is monotonic, we can consider $\varphi$ as an independent variable. Let's calculate $d\lambda/d\varphi$:
    \[
        \frac{d\lambda}{d\varphi} = 
        \frac{pr_x\cos2\varphi-\frac12(r_x^2 + \frac{\gamma^2}{\omega^2}- p^2)\sin2\varphi}{-pr_x\sin2\varphi - \frac12(r_x^2+\frac{\gamma^2}{\omega^2}-p^2)\cos2\varphi + \frac12(r_x^2+\frac{\gamma^2}{\omega^2}+p^2)}\,\lambda + O(\lambda^2),
    \]
    or
    \[
        \frac{d\lambda}{d\varphi} = 
        -\frac{A\cos2\varphi+B\sin2\varphi}{A\sin2\varphi - B\cos2\varphi - C}\,\lambda + O(\lambda^2),
    \]
    where $A=pr_x$, $B=-\frac12(r_x^2+\frac{\gamma^2}{\omega^2} - p^2)$, and $C=\frac12(r_x^2+\frac{\gamma^2}{\omega^2} + p^2)$. As mentioned earlier, $C^2 > A^2 + B^2$ as $p|_{t=0}\neq 0$.

    The key observation is that the derivative of the denominator of the fraction coincides with the numerator up to the factor $\frac12$ and a term $O(\lambda)$ of no importance:
    \[
        \frac{d}{d\varphi}(A\sin2\varphi - B\cos2\varphi - C) = 
        2(A\cos2\varphi+B\sin2\varphi) + (A'_\varphi\sin2\varphi - B'_\varphi\cos2\varphi - C'_\varphi)
    \]
    Indeed, functions $A(t)$, $B(t)$, and $C(t)$ are Lipschitz continuous. Hence $A'_\varphi=\dot A/\dot\varphi = O(\lambda)$, and similarly $B'_\varphi=O(\lambda)$ and $C'_\varphi=O(\lambda)$. The denominator $A\sin2\varphi-B\cos2\varphi-C$ is separated from $0$ as $\varphi\to+\infty$. So
    \[
        \frac{d\lambda}{d\varphi} = 
        -\lambda \frac{d}{d\varphi}\Big[\frac12\ln(-A\sin2\varphi + B\cos2\varphi + C)\Big] + O(\lambda^2)
    \]
    The expression under the logarithm is precisely $S^2(\theta_S)$. Thus,
    \[
        \frac{d\lambda}{d\varphi} = -\lambda \frac{d}{d\varphi}\ln S(\theta_S) + O(\lambda^2),
    \]
    
    Denote $\nu=\lambda S$. Then $\frac1c \nu \le \lambda \le c\nu$ as $\varphi\to+\infty$ and
    \[
        \frac{d\nu}{d\varphi} = S(\theta_S)\frac{d\lambda}{d\varphi} + \lambda\frac{dS(\theta_S)}{d\varphi} = S\cdot O(\lambda^2) = O(\nu^2).
    \]
    Since $\varphi\to+\infty$ and $\nu\to+0$ as $t\to+0$, we only need a lower bound estimate: $\frac{d\nu}{d\varphi}\ge -c\nu^2$, which guarantees that function $\nu(\varphi)$ cannot converge to zero too quickly as $\varphi\to+\infty$. Indeed, from this estimate, it follows that $\nu\ge\frac1c\frac1\varphi$ for large $\varphi$. Consequently, $\lambda\ge\frac1c\frac1\varphi$. Thus, $\dot\varphi\le \frac c\lambda\le c\varphi$, and therefore, the function $\varphi(t)$ is bounded for small $t$, which contradicts the asymptotic behavior $\varphi(t)\to+\infty$ as $t\to +0$.
    
\end{proof}

Thus, the optimal trajectory $(r_x(t),R(t),p(t),q(t))$ lies within the domain of analyticity of the function $\theta_M(r_x,R,p,q)$ for all $t\in(0;T]$. The optimal trajectory is a solution of the system of ordinary differential equations of the maximum principle and lies within the domain of analyticity of the right-hand side for all $t\ne 0$. Therefore, according to the Cauchy-Kovalevskaya theorem, its solution is analytic.

It remains to note that if $R(\tau) = 0$ at some instant $\tau$, then $q(\tau) \neq 0$. Therefore, according to the Pontryagin's maximum principle, the control $\theta(\tau)$ maximizes the following function
\[
    q\big(r_x\cos\theta + \frac\gamma\omega\cos\theta\big)\to\max_\theta
\]
Therefore $\cos\theta(\tau)=\sgn q(\tau) r_x/\sqrt{r_x^2(\tau) + \frac{\gamma^2}{\omega^2}}$, $\sin\theta(\tau)=\sgn q(\tau) \frac\gamma\omega/\sqrt{r_x^2(\tau) + \frac{\gamma^2}{\omega^2}}$, and hence $\dot R(\tau)=\sgn q(\tau)\omega\sqrt{r_x^2(\tau) + \frac{\gamma^2}{\omega^2}}$

\section{Proof of Proposition~\ref{prop:pureness}}
\label{sec:appendix_proof_prop:pureness}

Since $|R|\le 1$, it follows from formula ~\eqref{eq:pureness_derivative} that the maximum growth rate of $z=\sqrt{r_x^2+R^2}$ is achieved when $\theta=\frac\pi2 \sgn R$, and the minimum is achieved when $\theta=-\frac\pi2\sgn R$. Therefore,
\[
-\frac\gamma\omega (z^2+2|R|+R^2)\le
\frac1\omega \frac{\D}{\D t}z^2 \le -\frac\gamma\omega (z^2-2|R|+R^2)
\]
Estimates $|R|\le z\le 1$ imply
\begin{equation}
\label{eq:dot_z_estimate}
-\gamma(z+1) \le
\dot z\le
\gamma (1-z).
\end{equation}
The left and right equality is achieved with an appropriate choice of $\theta$ only if $|R|=z$, which implies $r_x=0$. It should be noted that if $\theta=\pm\frac{\pi}{2}$ and $r_x=0$, then $\dot{r}_x=0$. Thus, the optimal motions in the problem with endpoint constraints $|\mathbf{r}(0)|=\mu_0$, $|\mathbf{r}(T)|=\mu_1$ are movements along the axis ${r_x=0}$, where $\theta=\pm\frac{\pi}{2}$, and the sign is chosen to satisfy the equality in the left or right estimate of equation~\eqref{eq:dot_z_estimate}. Consequently, function $R(t)$ is monotonic, $|R(t)|=z(t)$, and $R(t)$ does not change sign. Without loss of generality, we assume that $R(t)\ge 0$. Therefore,
\[
	\theta = \begin{cases}
		\frac\pi2,&\text{if }\mu_0<\mu_1;\\
		-\frac\pi2,&\text{if }\mu_0>\mu_1.
	\end{cases}
\]
or $\theta=\sigma\frac\pi2$. Then $\dot z = \gamma (\sigma-z)$ and $|\sigma-z|=1-\sigma z = ce^{-\gamma t}$. From the condition $z(0)=\mu_0$, we obtain $c=1-\sigma \mu_0$, and from $z(T)=\mu_1$, we find $T=\frac{1}{\gamma}(\ln(1-\sigma \mu_0)-\ln(1-\sigma \mu_1))$.

\bibliographystyle{unsrturl.bst}

\bibliography{bib.bib}

\end{document}